%% file: TwoDimensional-density-classif.tex
\documentclass{article}

\usepackage{dsfont}
\usepackage{url}

\usepackage{tikz}

\include{cmd-math}

\include{cmd-CA}

\include{cmd-ferrero}

\include{cmd-figures-trafic2D}
\graphicspath{{Figures/}}

\newtheorem{proposition}{Proposition}
\newcommand{\QED}{}
\newtheorem{lemma}{Lemma}

\newcommand{\ad}[2]{\footnotesize{#1}\\ \, \footnotesize\url{#2}\\ }

\title{Two-dimensional traffic rules \\ and the density classification problem}

\author{Nazim Fat\`es\\
\ad{Inria Nancy - Grand Est LORIA UMR 7503}{nazim.fates@loria.fr}
\\ Ir\`ene Marcovici\\
\ad{Institut \'Elie Cartan de Lorraine, Universit\'e de Lorraine}{irene.marcovici@univ-lorraine.fr}
\\ Siamak Taati\\
\ad{Mathematics Institute, Leiden University}{siamak.taati@gmail.com}
}

\begin{document}

\maketitle

\begin{abstract}
\input{abstract.txt}
\end{abstract}

\input{appendix} 
\input{body}

\section*{Acknowledgments}
The work of Siamak Taati is supported by ERC Advanced Grant 267356-VARIS of
Frank den Hollander.

\bibliographystyle{plain}
\bibliography{bib-CA,bib-Fates}

\end{document}

%% file: cmd-math.tex
\usepackage{amsmath}
\usepackage{amsfonts}
\usepackage{amssymb}
\usepackage{amsthm} 

\newcommand{\Z}{\mathbb{Z}} 			
\newcommand{\N}{\mathbb{N}} 			
 
\newcommand{\LL}{{\mathcal L}} 			

\def \card {\mathop{\rm Card }\nolimits}

\newcommand{\fleche}{\rightarrow}		
\newcommand{\function}[3]{ {#1} : {#2} \fleche  {#3} } 	

\usepackage{textcomp}

\newcommand{\set}[1]{\{#1\}}

\newcommand{\und}[1]{_{\mathrm{#1}}}

\newcommand{\TextIf}{\mbox{ if }}


%% file: cmd-CA.tex
\newcommand{\QQ}{\set{\qO, \qX}}

\def \A {\mathcal{A}}

\newcommand{\Neighb}{{\mathcal N}}
\newcommand{\V}{{\mathcal V}} 	

\newcommand{\etat}[1]{{\ensuremath{\mathtt{#1}}}}

\newcommand{\qO}{{\etat{0}}}
\newcommand{\qX}{{\etat{1}}}



%% file: cmd-ferrero.tex
\newcommand{\II}{\mathcal I}

\newcommand{\DD}{\mathcal D}

\newcommand\ind{\mathds{1}}

\newcommand{\EE}{\mathcal{E}}

\newcommand\Ni{\Neighb\und{i}}
\newcommand\Np{\Neighb\und{p}}

%% file: cmd-figures-trafic2D.tex
\newcommand{\drawpattern}[1]{%
	\begin{tikzpicture}
		\foreach \line [count=\y] in #1 {
			\foreach \pix [count=\x] in \line {
      			\draw[fill=pixel \pix] (0.3*\x,-0.3*\y) rectangle +(0.3,0.3);
    		}
  		}
	\end{tikzpicture}
}
\definecolor{pixel 0}{HTML}{0000FF}
\definecolor{pixel 1}{HTML}{FFFFFF}

\def\pixels{
  {1,1,0,1,0,1,0,1,0,0,1,0,1,0,1,0},
  {0,1,1,0,1,0,1,0,1,0,0,1,0,1,0,1},
  {1,0,1,1,0,1,0,1,0,1,0,0,1,0,1,0},
  {0,1,0,1,1,0,1,0,1,0,1,0,0,1,0,1},
  {1,0,1,0,1,1,0,1,0,1,0,1,0,0,1,0},
  {0,1,0,1,0,1,1,0,1,0,1,0,1,0,0,1},
  {1,0,1,0,1,0,1,1,0,1,0,1,0,1,0,0},
  {0,1,0,1,0,1,0,1,1,0,1,0,1,0,1,0},
  {0,0,1,0,1,0,1,0,1,1,0,1,0,1,0,1},
  {1,0,0,1,0,1,0,1,0,1,1,0,1,0,1,0},
  {0,1,0,0,1,0,1,0,1,0,1,1,0,1,0,1},
  {1,0,1,0,0,1,0,1,0,1,0,1,1,0,1,0},
  {0,1,0,1,0,0,1,0,1,0,1,0,1,1,0,1},
  {1,0,1,0,1,0,0,1,0,1,0,1,0,1,1,0},
  {0,1,0,1,0,1,0,0,1,0,1,0,1,0,1,1},
  {1,0,1,0,1,0,1,0,0,1,0,1,0,1,0,1},
}

\newcommand{\figProof}{
\begin{figure}[h]
\centering
\drawpattern{\pixels}
\caption{Example of a configuration that locally looks like an archipelago.}
\label{fig:exCW}
\end{figure}

}

\newcommand{\figHomoPhase}{
\begin{figure}
\includegraphics[width=.48\linewidth]{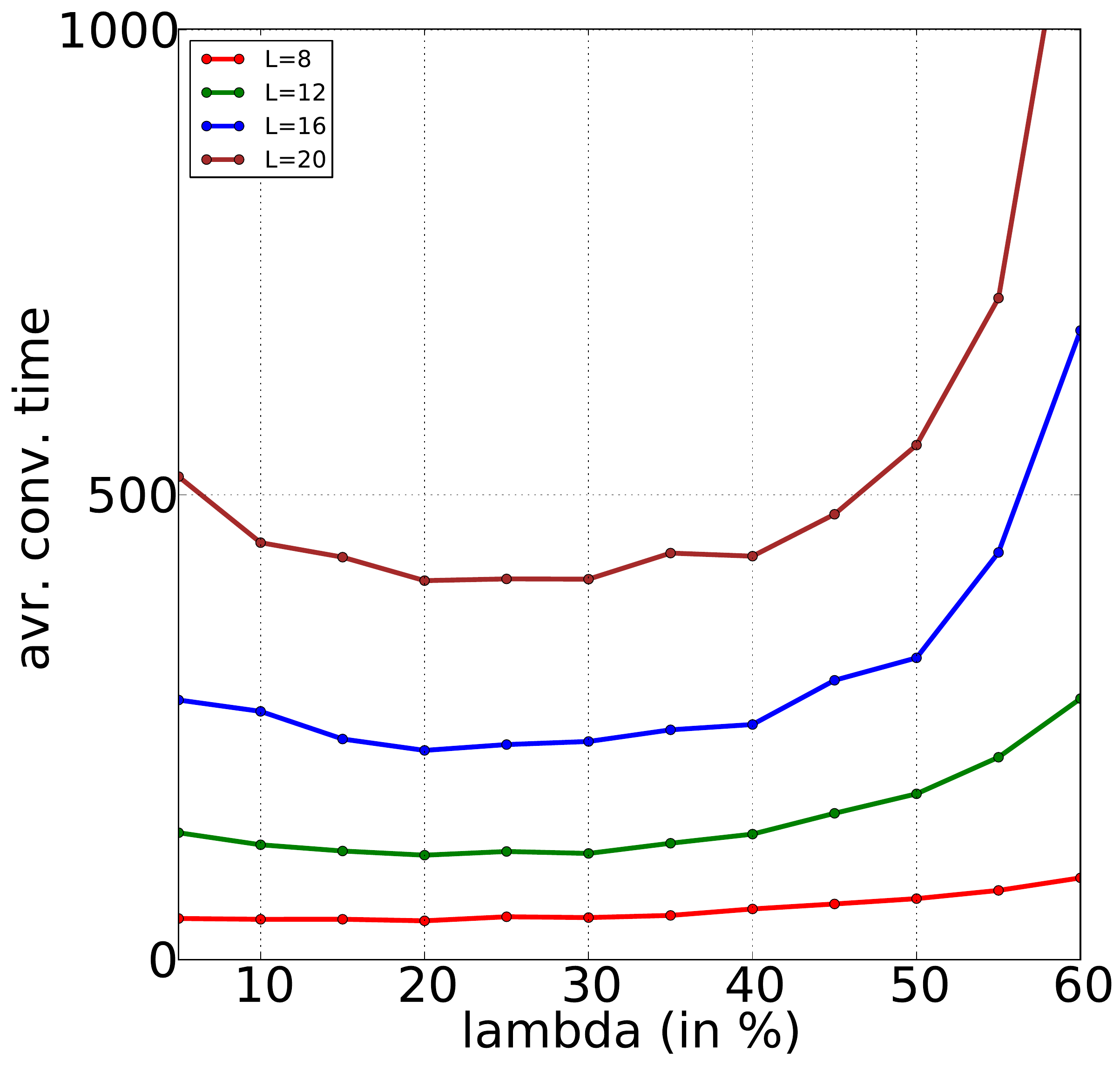}
\,
\includegraphics[width=.48\linewidth]{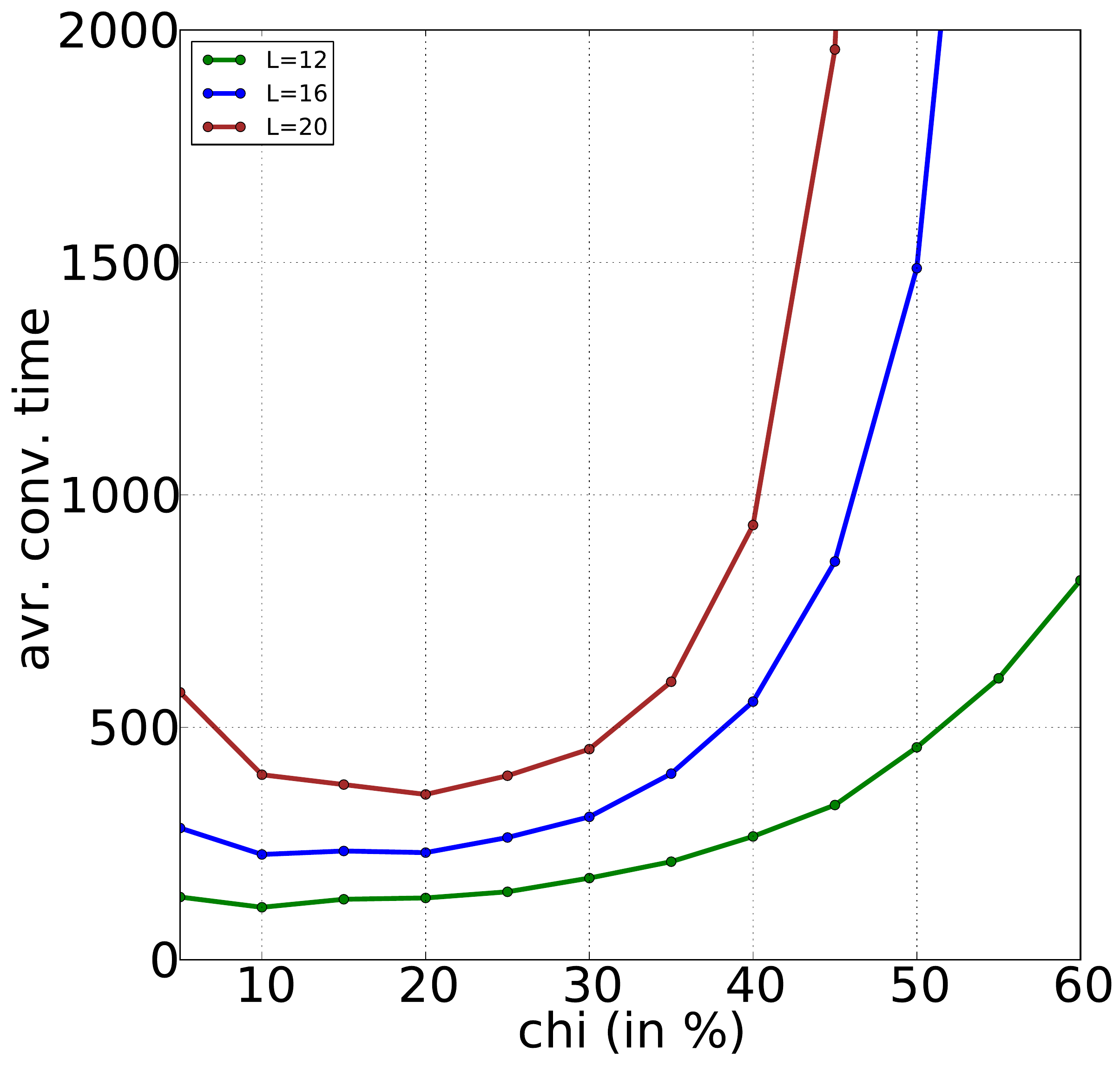}
\caption{Time to convergence to an archipelago for different values of $ L $.
Left: $ \lambda $ varies with $\chi=0.1$.
Right: $ \chi $ varies with $\lambda=0.25$
}
\label{fig:homo}
\end{figure}
}

\newcommand{\figK}{
\begin{figure}
\centering
\includegraphics[width=.7\linewidth]{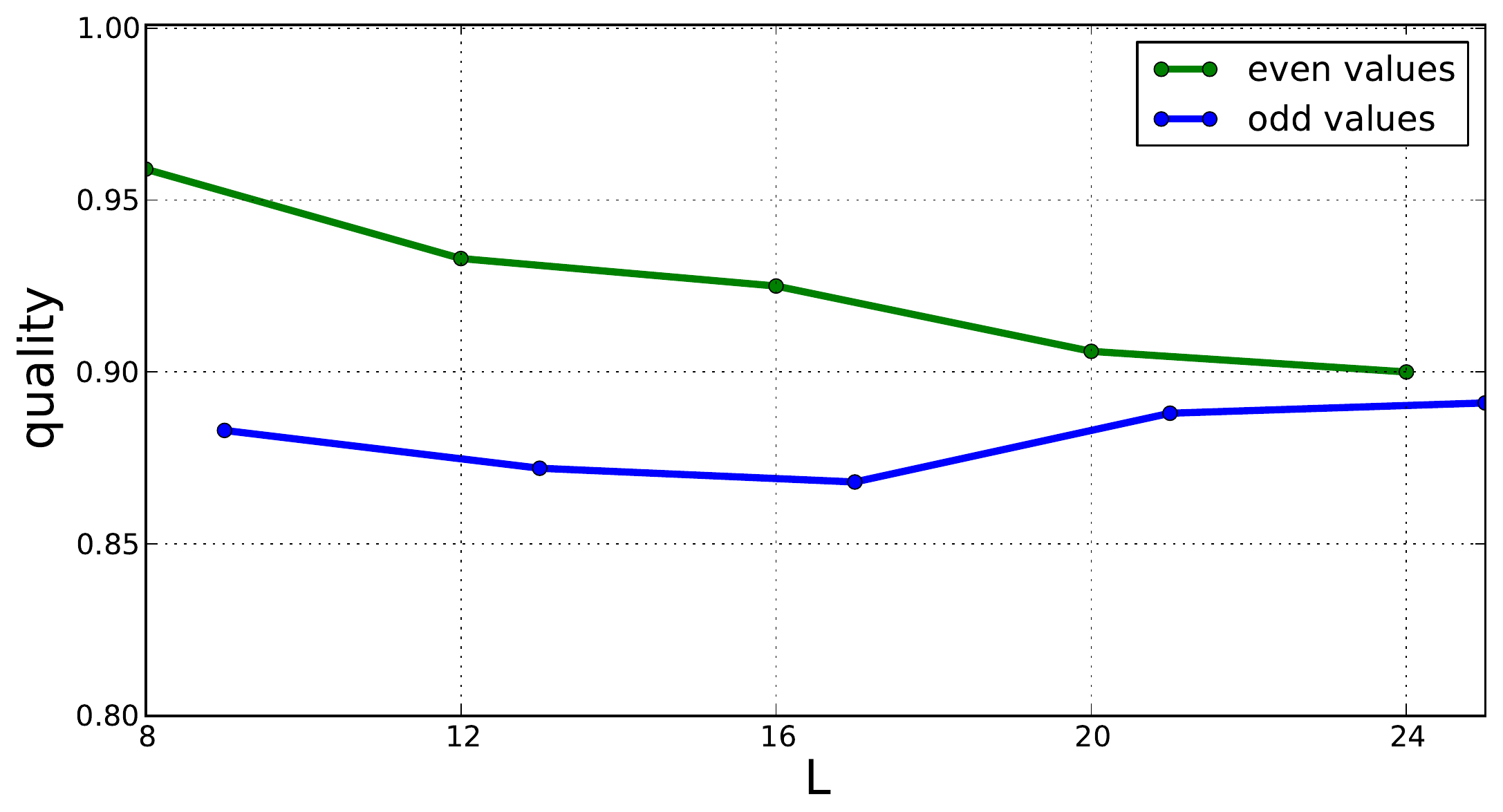}
\caption{Quality of classification as a function of the grid width $ L $. 
(Even and odd values are shown on separate curves.) Settings are: 
$ \lambda =0.25$, $ \chi = 0.1$, $ \epsilon = 0.001$.
}
\label{fig:koupa}
\end{figure}
}

\newcommand{\cSnap}[1]{
\begin{tabular}{c}
\includegraphics[width=.2\textwidth]{fig/snap-homophase-t-#1.pdf}\\
\end{tabular}
}

\newcommand{\figSnapshotsC}{
\begin{figure}
\noindent
\begin{tabular}{c c c c}
\cSnap{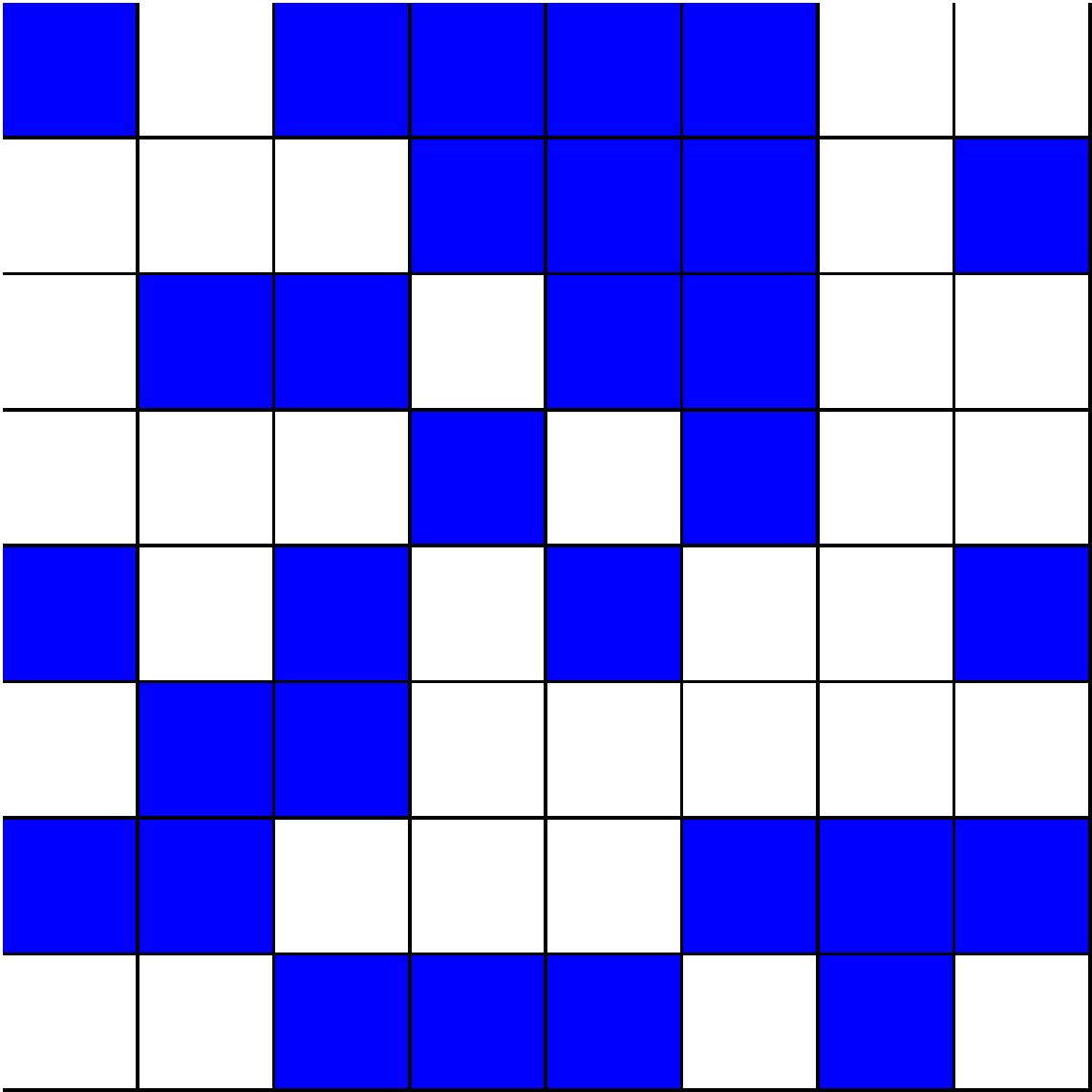}&\cSnap{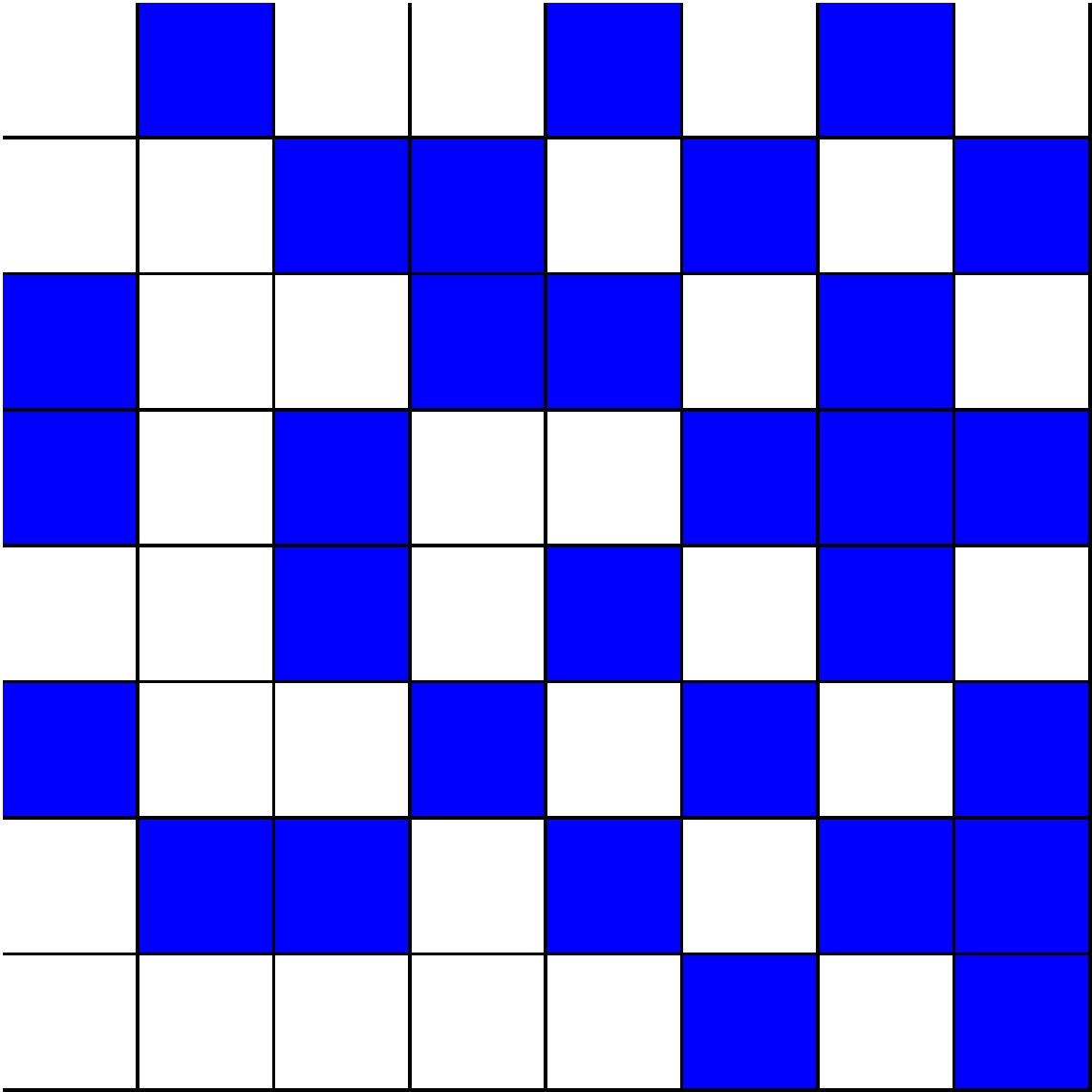}&\cSnap{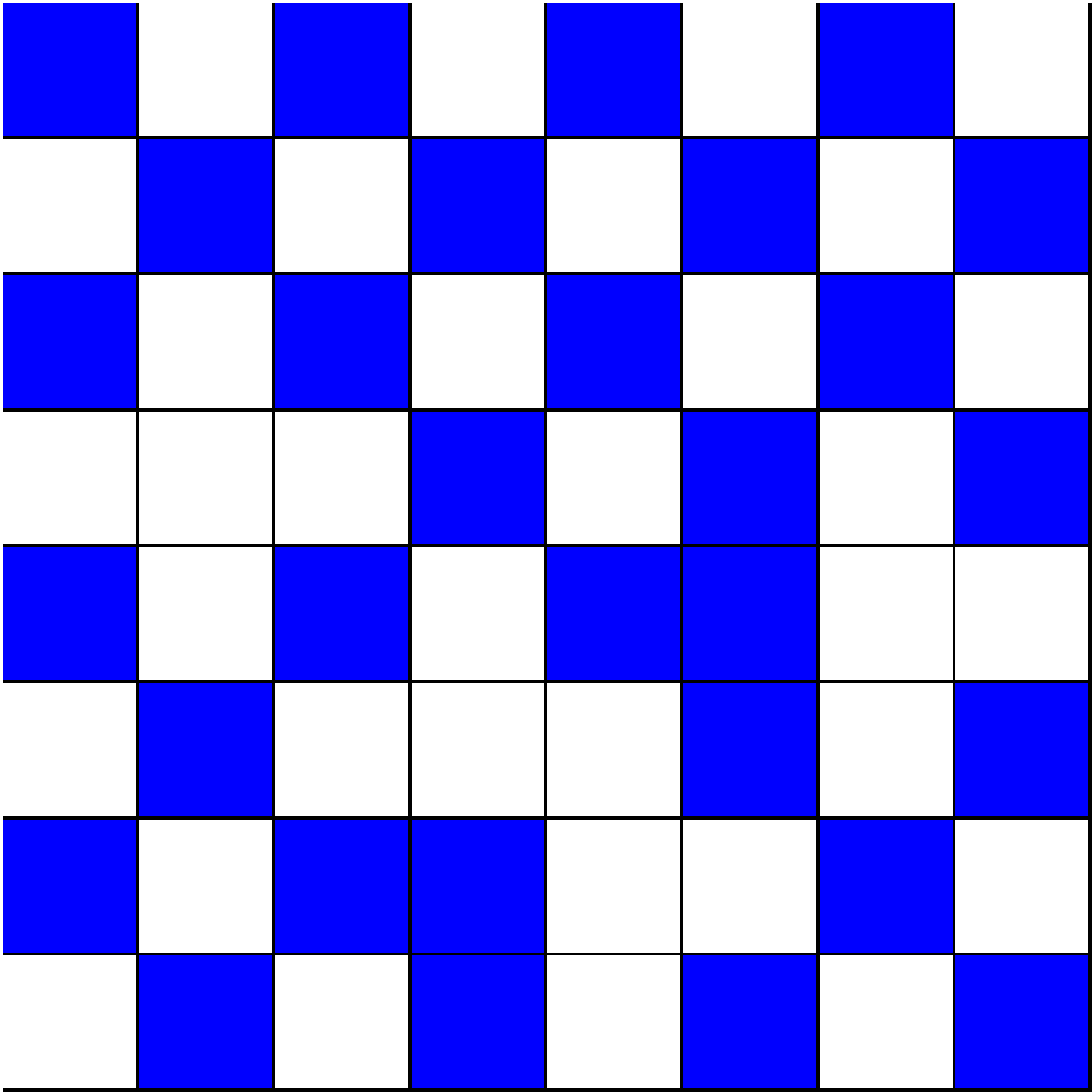}&\cSnap{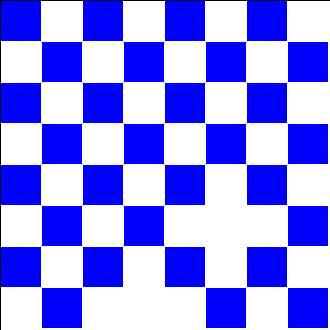}\\
\end{tabular}
\caption{Transformation of a random initial configuration into an archipelago.
}
\label{fig:snapC}
\end{figure}
}

\newcommand{\dSnap}[1]{
\begin{tabular}{c}
\includegraphics[width=.2\linewidth]{fig/snap-dc-t-#1.pdf}\\
\end{tabular}
}
\newcommand{\figSnapshotsD}{
\begin{figure}
\begin{tabular}{c c c c}
\dSnap{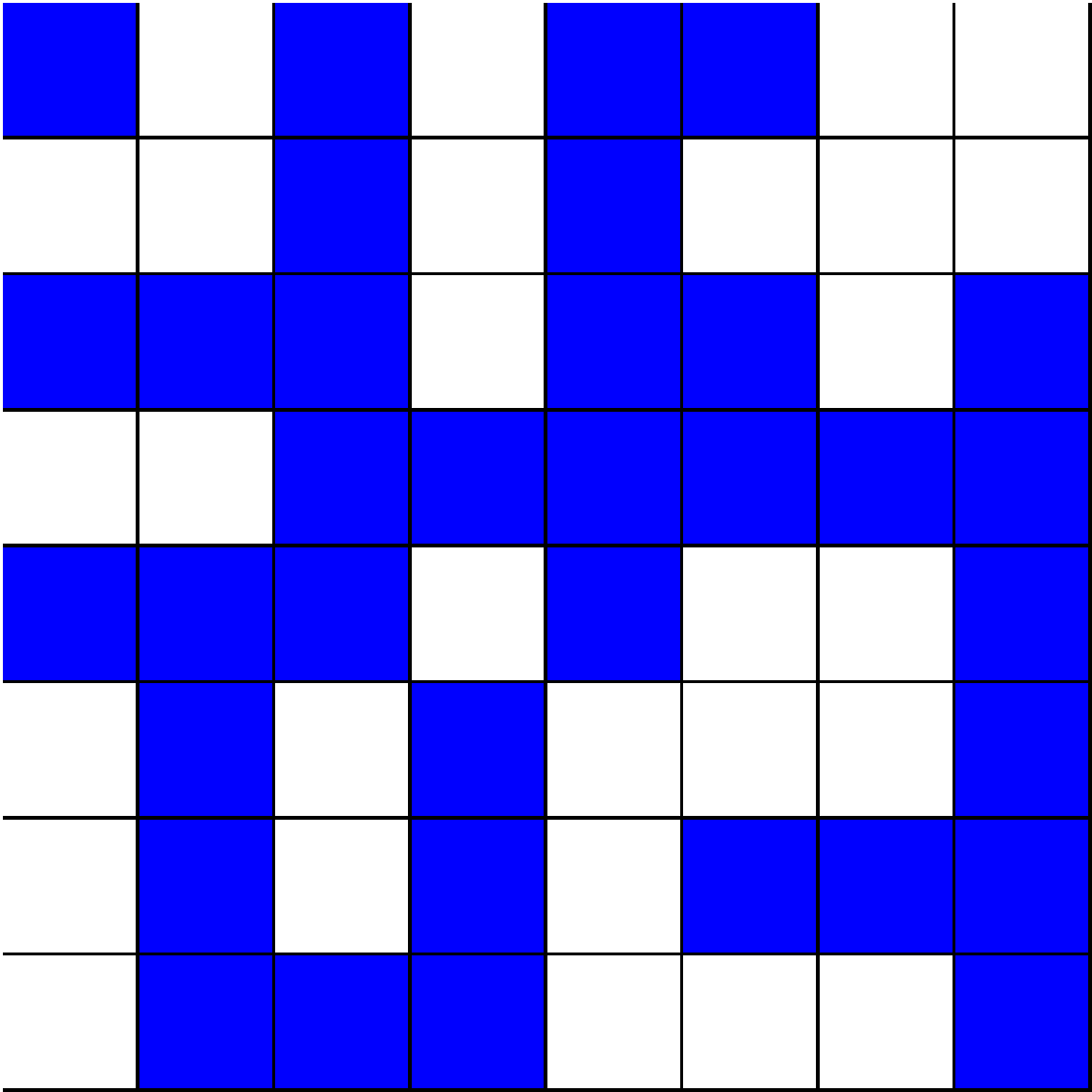}&
\dSnap{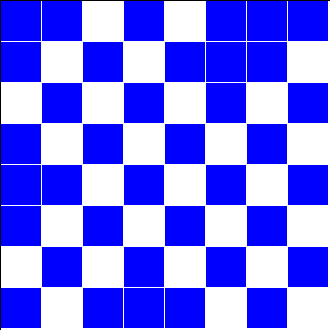}&
\dSnap{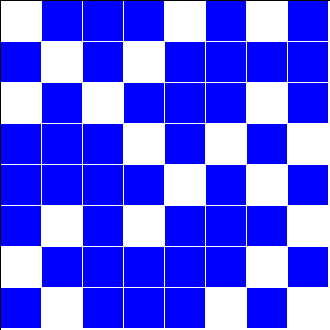}&
\dSnap{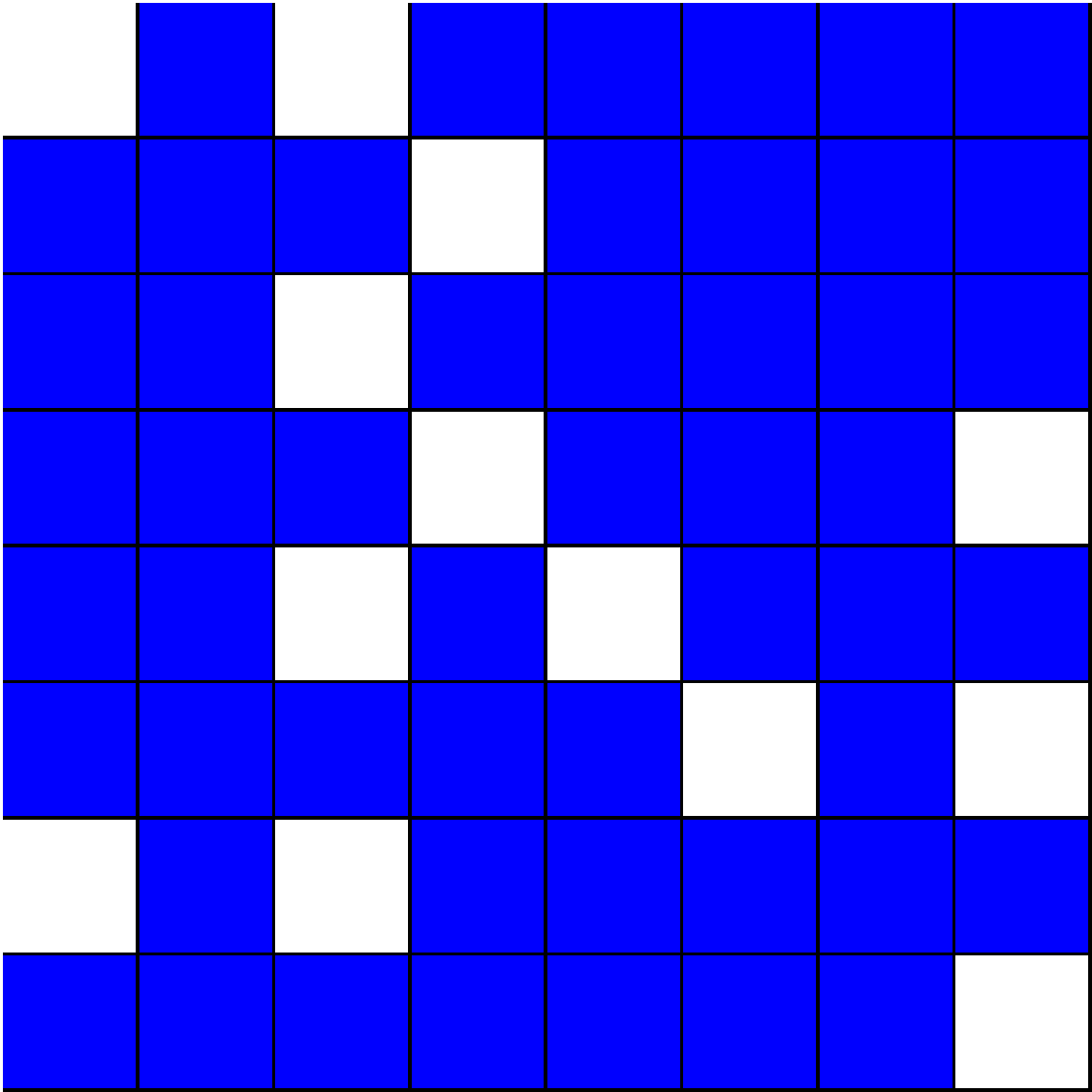}\\
\end{tabular}
\caption{
Transformation of a random inital condition. Under some conditions (see Sec.~\ref{sec:densClassifRule} p.~\pageref{sec:densClassifRule}), the system will most probably converge to $ \qX^\LL$, which is a correct classification since the initial density is greater than 1/2.
}
\label{fig:snapD}
\end{figure}
}


%% file: abstract.txt
The density classification problem is the computational problem of
finding the majority in a given array of votes, in a distributed
fashion. It is known that no cellular automaton rule with binary
alphabet can solve the density classification problem.  On the other hand,
it was shown that a probabilistic mixture of the traffic rule and the majority
rule solves the one-dimensional problem correctly with a probability
arbitrarily close to one.  We investigate the possibility of a similar
approach in two dimensions.  We show that in two dimensions, the
particle spacing problem, which is solved in one dimension by the traffic rule,
has no cellular automaton solution. However, we propose exact and randomized
solutions via interacting particle systems.  We assess the performance
of our models using numeric simulations.

%% file: appendix.tex
\newcommand{\ProofSpacingLoose}{
\begin{proof}[Proposition~\ref{prop:spacing:loose:solution}]
It is clear that the IPS $\PhiC$ constructed above is conservative.
We show that the set $C$ of sub-checkerboards is a sink for $\PhiC$.

First, note that the set $C$ is absorbing.
Indeed, in every element of $\Ce{q}\cup \Co{q}$, every $q$ is isolated,
and as a result $\PhiC(x,(k,k'))=x$ for each $x\in C$ and $(k,k')\in\II_4$
(i.e., the horizontal/vertical rules keep $x$ unchanged).
On the other hand, for each $x\in C$ and $(k,k')\in\DD_4$,
clearly we have $\PhiC(x,(k,k'))\in C$
(i.e., $C$ is invariant under the diagonal rule).
It remains to verify that $C$ is reachable from any configuration.

Let $x\in\EE$ be a configuration, and without loss of generality assume that $d_{\qX}(x)\leq 1/2$.
Let us denote the number of even cells $(i,j)$ with $x_{i,j}=\qX$ by $m_e$
and the number of odd cells $(i,j)$ with $x_{i,j}=\qX$ by $m_o$.
If $m_e=0$ or $m_o=0$, then $x$ is already in $C$, and there is nothing to prove.
So, assume $m_e$ and $m_o$ are both non-zero.
We describe a path from $x$ to a configuration $y$ that has $m_e+1$ even cells in state $\qX$.
The claim would then follow by induction.

Let $M=\{0,1\}^2\subseteq\LL$ be a window consisting of a $2\times 2$ square.
We first show that there is a position $c$ such that the $2\times 2$ pattern $(x_{c+k}: k\in M)$ contains either a single $\qX$ or two adjacent $\qX$s and two adjacent $\qO$s (see Fig.~\ref{fig:spacing:loose:illustration}). 
Indeed, there are three possibilities:
\begin{enumerate}
	\item Every $2\times 2$ window has at least $2$ occurrences of $\qX$.
		In this case, every $2\times 2$ window must have precisely two occurrences of $\qX$, for otherwise $d_{\qX}(x)>1/2$.
		If every $2\times 2$ window is "checkered", then $x$ must be in $\Ce{\qX}\cup \Co{\qX}$, which contradicts the assumption.
		Otherwise, $x$ has a $2\times 2$ pattern with two adjacent $\qX$s and two adjacent $\qO$s.
	\item There is a $2\times 2$ window with a single $\qX$.
	\item There is a $2\times 2$ window with no $\qX$s.  Let us call such a window \emph{empty}.
		Note that every window that either horizontally or vertically overlaps an empty window is either itself empty or has one of the two desired forms.
		Since we have assumed that $x$ has at least one~$\qX$ ($m_e,m_o>0$), not every $2\times 2$ window can be empty.
		Consider a path from an empty window to a non-empty window, obtained by vertical and/or horizontal moves.
		The first time such a path reaches a non-empty window, we see a pattern with one of the two desired forms.
\end{enumerate}

Let us consider a $2\times 2$ window $c+M$ that has a single $\qX$.
Without loss of generality, we can assume that this single $\qX$ is on an even cell.
Since $m_o>0$, there is at least one $\qX$ somewhere outside $c+M$ on an odd cell.
By applying the diagonal rule, we can move this single $\qX$ to $c+M$,
without changing any other cell, hence obtaining a new configuration that differs from $x$ in exactly two cells.
In this new configuration, the window $c+M$ has two adjacent $\qX$s and two adjacent~$\qO$s.

Consider finally a $2\times 2$ window $c+M$ with two adjacent $\qX$s and two adjacent~$\qO$s.
Let $k$ be the odd cell in $c+M$ that has a $\qX$ and $k'$ be its adjacent cell in $c+M$ having state~$\qO$.
Applying the interaction rule on pair $(k,k')$ we end up in a configuration $y$ with one more even cell in state~$\qX$,
hence concluding the proof.\QED\end{proof}
}


\newcommand{\ProofLemmaEns}{
\begin{proof}[of Lemma~\ref{lem:ens}]
For $x\in\EE$ and $q\in \QQ$, let us set 
$E^q(x)=\sum_{(i,j)\in\II_4} \ind[ x_i= x_j = q ].$

For a configuration $x\in\EE_k$, if we count in two different ways the number of pairs $u\in \II_4$ with one cell in state $\qO$ and the other in state $\qX$, we obtain the equality $\card \II_4-E(x)=4k- E^{\qX}(x)$. The energy $E(x)$ is thus minimal within $\EE_k$ if and only if $E^{\qX}(x)$ is minimal. In particular, if there are configurations $x\in\EE_k$ such that $E^{\qX}(x)=0$ (which is equivalent to $x\in\A_{\qX}$), then these are exactly the configurations of minimal energy. Similarly, if there are configurations $x\in\EE_k$ such that $E^{\qO}(x)=0$ (which is equivalent to $x\in\A_{\qO}$), then these are exactly the configurations of minimal energy. Since the grid is even-sized, one can check that whatever the value of $k$, the set $\EE_k\cap \A$ is non-empty. And we have $\EE_k\cap \A=\EE_k\cap\A_q,$ where $q$ is the minority state (or any state in case of equality). 
\QED\end{proof}
}


\newcommand{\ProofPropExistence}{
\begin{proof} By contradiction, assume that we have an IPS rule that is a solution to the strict spacing problem. 
For $c\in\Z^2$, and $r\geq 1$, let us introduce the notation $B(c,r)=\{c+\delta \, ; \; ||\delta||_1\leq r\}$.

The interaction and the perception neighbourhoods being bounded, there exist $r\geq 1$ such that $\Np\subset B(0,r)$, and $\Delta\geq 1$ such that for any $(c,c')\in\II$, the set $B(c,r)\cup B(c',r)$ has a diameter smaller than $\Delta$.

Let us consider the lattice $\LL=(\Z/n\Z)^2$, with $n=4k$, for some $k>\Delta$. 
We define the configuration $x\in\EE$, by:
$x_{i,j}= \qO $ if and only if 
$(j-i)\in\{0\}\cup\{2i+1 ; 0\leq i\leq k-1\}\cup\{2k+2i ; 1\leq i\leq k-1\}$.
This configuration is represented on Figure~\ref{fig:exCW} for $k=4$. It is not an archipelago. Moreover, observe that every $\qX$ has at least two adjacent cells in state $\qO$, and the other way round. 
Note that this configuration has the same density of $\qO$s and $\qX$s, but it is also possible to build other counter-examples with various densities.

Let $(c,c')\in\II^2$ be such that $x_c\neq x_{c'}$. 
By construction of $x$, since $k>\Delta$, one can check that on the set of cells $B(c,r)\cup B(c',r)$, either there are no pairs of adjacent $\qX$s or there are no pairs of adjacent $\qO$s. Let us for example assume that there are no pairs of adjacent $\qX$s.

We now consider the configuration $y$ that coincides with $x$ on the set $B(c,r)\cup B(c',r)$ and for which all the cells outside this set are in state $\qO$. The configuration $y$ is a $\qX$-archipelago. For configuration $y$, the exchange between the states of cells $c$ and $c'$ is forbidden, since it would create a pair of adjacent $\qX$s (because among $c$ and $c'$, the cell in state $\qX$ has at least two adjacent cells in state $\qO$).

Let us go back to $x$. Since $x$ and $y$ are locally the same, the exchange of $c$ and $c'$ is forbidden. This implies that configuration $x$ is a fixed point of the dynamics, which is a contradiction, because it is not an archipelago.

\QED\end{proof}
}

%% file: body.tex
\newcommand{\PhiC}{\Phi\und{C}}
\newcommand{\PhiT}{{\tilde{\Phi}}\und{C}}
\newcommand{\PhiD}{\Phi\und{D}}

\section{Introduction}

Let us imagine a medium composed of a great number of cells arranged regularly on a grid. Each cell is linked with its immediate neighbours and the only thing it can do is to change its own state according to the state of its neighbours. Can we compute with such a medium? And what happens if the updates occur at random times? And what if the cells are subject to noise?

%
In order to study this robustness mechanisms on a mathematical basis, we will here focus on two simple computational problems.
The first problem is the \emph{density classification problem}, which is the problem of finding the majority state in a distributed fashion.
In the original setting, the computational protocol is required to be local and parallel, and use no extra memory other than the evolving configuration itself.
The protocol is also required to be scalable, which means it must perform the task on configurations of arbitrary size.
In other words, we look for a cellular automaton rule that performs the task.
In this paper, we also consider variants of this problem in which the process is allowed to be asynchronous, non-deterministic or random.

This problem has attracted a considerable amount of attention these last years. 
It is trivial in most settings but it is not easy to solve in the case of cellular automata. 
The difficulty comes from the necessity to reach a consensus on the state of the cells: the system should converge to a situation with all 1s or all 0s, depending on whether the initial state contains more 0s or more 1s, respectively.

Inspired by the work of G\`acs, Kurdiumov and Levin, in 1988, Packard formulated this problem as a challenge to study genetic algorithms~\cite{Pac88}.
This triggered a wide competition to find rules with an increasing quality of classification. 
In 1995, Land and Belew proved that no perfect solution exists for one-dimensional deterministic systems~\cite{Lan95}. 
Recently, this fact was re-demonstrated with a simpler argument and the proof was extended to probabilistic rules and to any dimension~\cite{BFMM13}. It was even shown that for any candidate solution there are configurations with a density close to 0 and 1 that are misclassified~\cite{KG12}.

Since then, different variants of the problem have been proposed and it has been shown by various authors that relaxing one of the conditions of the problem is often sufficient to find perfect solutions~\cite{Oli14}. In particular, Fuk\'s proposed to combine two rules sequentially to obtain a perfect solution, see Ref.~\cite{Fuk15} and references therein.
Probabilistic cellular automata could provide another interesting framework: it was discovered that although no perfect rule exists, it is possible to find a family of one-dimensional nearest-neighbour rules for which the probability of making an error of classification can be made as small as wanted~\cite{FatTocs13}. The perfect solution can thus be approximated -- but not reached ! -- at the cost of an increase in the average time to reach a consensus.

The construction proposed for building this family of rules consists of mixing stochastically two well-known rules: the traffic rule, which introduces space between particles, and the majority rule, which has a ``homogenising" effect.
In this text, we ask whether there also exist a ``close-to-perfect" solution for two-dimensional cellular automata. 
At first sight, one does not see why the problem should be significantly different for two-dimensional systems. However, there is no such thing as a ``traffic" rule in two dimensions (2D). If we decompose a 2D grid in layers and  apply a classical traffic rule on each layer, then different consensuses might be attained and there is no obvious means on how one can obtain the ``right" global consensus from a collection of local consensuses.

We call the problem that is solved by the traffic rule in one dimension the {\em particle spacing problem}. 
We tackle this problem in two dimensions.
We then (partially) solve the density classification problem by combining our particle spacing model with a local majority rule.

The outline of the article is as follows. 
After presenting the basic definitions and properties of our models in Sec.~\ref{sec:basis},
we show the advantage of using interacting particle systems to tackle the problem (Sec.~\ref{sec:ipsAnalysis}).
We then present a concrete solution and 
analyse its behaviour with numerical simulations in Sec.~\ref{sec:densClassifRule}.

\section{Basics}
\label{sec:basis}

\subsection{Setting}

In dimension $d\geq 1$, we set the cellular space to be a grid with periodic boundary conditions, defined by $\LL = (\Z/n_1\Z) \times\cdots\times (\Z/n_d\Z)$, for some $n_1,\ldots,n_d\geq 1$. The number of cells of $\LL$ is $N_{\LL}=n_1\cdots n_d$. We say that the grid is \emph{even-sized} if $n_1,\ldots, n_d$ are all even. 

Each cell of this space can hold a binary state, so that the set of states is denoted by $ Q = \set{ \qO, \qX } $.

The set of configurations is denoted by $\EE=\{\qO,\qX\}^{\LL}$. 

For a configuration $x\in \EE$, and a state $q\in Q$, we define the density of state~$q$ by:
$d_q(x)={1\over N_{\LL}}\card\{i\in \LL \, ; \, x_i=q\}.$

For a given configuration $x\in \EE$, we say that cell $i$ is \emph{isolated} if none of its adjacent cells is in state $x_i$.

For $q\in Q$, we say that a configuration $x\in \EE$ is a \emph{$q$-archipelago} if all the cells in state $q$ are isolated, i.e., if $x$ does not contain two adjacent cells in state $q$. We denote by $\A_q$ the set of $q$-archipelagos. In particular, if $x\in \A_q$, then $d_q(x)\leq 1/2$.

We also introduce $\A=\A_{\qO}\cup\A_{\qX}$, the set of all archipelagos.

\figSnapshotsD

\subsection{Presentation of the problem}

Recall that in this paper, we study two computational problems.
The \emph{density classification task} is the task of transforming a given configuration $x\in Q^\LL$ into one of the two uniform configurations $\qX^\LL$ or $\qO^\LL$ depending on which of $\qX$ or $\qO$ has strict majority in $x$ (see Fig.~\ref{fig:snapD}).
More specifically, given an input $x\in Q^\LL$, a computational process performing the density classification task must return $\qX^\LL$ if $d_{\qX}(x)>1/2$ and $\qO^\LL$ if $d_{\qO}(x)>1/2$.
(Generally, the case $d_{\qX}(x)=d_{\qO}(x)$ is avoided.)

Our approach to solve the density classification problem is via another problem which we call the particle spacing problem.
The \emph{particle spacing problem} is the computational problem of rearranging the ``particles'' (say, symbols $\qX$) on a configuration $x\in Q^\LL$ so as to obtain an archipelago configuration (see Fig.~\ref{fig:snapC}).
Again we require the computational process to be local and scalable, but we also require it to be \emph{conservative}: at every step of the process, the number of particles (symbols~$\qX$) must be preserved.

There are two possible variants for the latter problem.
In the \emph{strict} spacing problem, we require that the sets $\A_{\qX}$ and $\A_{\qO}$ are \emph{absorbing}, in the sense that as soon as the process enters $\A_q$ (for either $q=\qX$ or $q=\qO$), it cannot leave it.
In the \emph{loose} variant of the problem, we require the process to eventually remain in~$\A_q$.
Note that the latter is equivalent to the condition that the computation reaches an absorbing subset of $\A_{\qX}$ or $\A_{\qO}$.


For both problems, our purpose is to build a solution with a cellular system; this means that we have a set of interacting components, which can have a deterministic or stochastic behaviour, and interact only locally. We will here consider cellular automata (deterministic or stochastic) and interacting particle systems.

\subsection{Known results with cellular automata}

\figSnapshotsC

We now present the principal known results concerning solutions of the density classification problem and of the particle spacing problem, using cellular automata.

A cellular automaton $\function{F}{\EE}{\EE} $ is defined by a neighbourhood $\Neighb = (v_1, \dots, v_k) \in \LL^k$, and by a local rule $\function{f}{Q^k}{Q}$, which defines the global rule 
$F$, mapping a configuration $x$ to the configuration $F(x)$ defined by:
\[ 
\forall c \in \LL, F(x)_c= f\big( x_{c+v_1}, \dots, x_{c + v_k} \big) .
\]

\begin{proposition}
For any $d\geq 1$, there is no deterministic cellular automaton solving the density classification problem.
For $ d= 1$, this means that there is no local rule $f$ such that for any $n\geq 1$, and any $x\in \QQ^{\Z/n\Z}$, 
$$d_q(x)>1/2 \implies \exists T\geq 0, \forall t\geq T, F^t(x)=q^{\Z/n\Z}.$$
\end{proposition}

The first proof was given by Land and Belew in 1995 for dimension $d=1$~\cite{Lan95}. 
A simplified proof was proposed in 2013 for any dimension $d\geq 1$~\cite{BFMM13}. These results apply to deterministic cellular automata. One can ask whether stochastic transition rules could help to solve the problem. Using the same argument as for deterministic cellular automata, one can prove that there are no probabilistic cellular automata solving perfectly the density classification problem~\cite{BFMM13}.

However, for $d=1$, Fat\`es has provided a family of probabilistic cellular automata solving the density classification problem with an arbitrary precision~\cite{FatTocs13}. This means that the probability of making a bad classification can be reduced to as low as necessary, at the cost of an increase of the average time of convergence to the uniform configuration.

The family of rules is defined with a real parameter $\epsilon>0$. 
The local rule consists at each time step, for each cell independently, in applying the traffic rule with probability $1-\epsilon$ and the majority rule with probability~$\epsilon$. 
The traffic rule (rule 184 with Wolfram's notations) is a conservative rule, which moves the $\qX$s to the right whenever possible. 
It has a spacing effect. The majority rule allows the convergence to the uniform fix point, once particles have been spaced.

In order to extend this result to higher dimensions, one would like to design a rule having the same behaviour as the traffic rule, that is, to be able to compose a rule that solves the spacing problem with a majority rule. 
Unfortunately, this is not possible in the classical framework of cellular automata.

\begin{proposition} \begin{enumerate}
\item In dimension $1$, the traffic cellular automaton $F_{184}$ solves the spacing problem. 
Precisely, it satisfies: for all $n\geq 1$ and all $x\in \QQ^{\Z/n\Z}$,
$$d_q(x)\leq 1/2 \implies \forall t\geq n/2, \; F_{184}^t(x)\in \A_{q}.$$
Furthermore, $\forall q\in\QQ, F_{184}(\A_q)\subset \A_q$, so $F_{184}$ is a solution to the strict spacing problem.
\item In dimension $d\geq 2$, there are no cellular automata that solve the spacing problem.
\end{enumerate}
\end{proposition}

\begin{proof} The fact that the traffic cellular automaton spaces configurations is a ``folk'' result.

Let now $F$ be a $d$-dimensional deterministic cellular automaton. 
If $ x \in \EE $ is a configuration with a symmetry of translation,
then the symmetry is conserved by the evolution of the automaton.
Formally, if there exists $ \delta \in \Z^d$ such that
$ \forall c \in \LL, \, x_c= x_{c+\delta} $, then
$
\forall t \in \N, \, \forall c \in \LL, \, F^t(x)_c = F^t(x)_{c+\delta} 
$.
As a consequence, deterministic cellular automata can not solve the particle spacing problem in dimension $d\geq 2$. To see why, simply consider a configuration with all $\qO$s, except one line which is made of cells with all $\qX$s: if the rule is conservative, this line can not disappear.\QED\end{proof}

By its very nature, a ``truly'' probabilistic rule can not solve the particle spacing problem. Indeed, as soon as there exists a configuration for which one cell has a non-deterministic outcome, we cannot ensure that the number of particles will be preserved.
We leave open the question as to whether there exists a probabilistic rule which would solve the density classification problem with an arbitrary precision, in dimension 2 or more.

\section{Particle systems solutions to the spacing problem}\label{sec:IPS}
\label{sec:ipsAnalysis}

We have seen that deterministic cellular automata are in some sense too rigid to allow us to solve the spacing problem on grids, because they do not allow to break translation symmetries. On the other hand, probabilistic cellular automata can break these symmetries, but they do not allow an exact conservation of the number of particles. 

We now propose to combine the strength of both models with interacting particle systems: we update cells by \emph{pairs}, which allows conservation of particles, and the pairs are chosen \emph{randomly}, which allows us to break symmetries. The effect of the local rule is to exchange the cell's states or to leave them unchanged, depending on the states of the neighbouring cells of the pair. 

Let us formalize the definition of the \emph{interacting particle systems} (IPS) we consider. From now on, we will consider two-dimensional grids. Note that most results can be adapted to higher-dimensional lattices.

\subsection{Our model of IPS}

Let $\Neighb\und{i}$ and $\Neighb\und{p}$ be two finite tuples of $\Z^2$, corresponding to the \emph{interaction neighbourhood} and the \emph{perception neighbourhood}. 

We define the set of interacting pairs by 
$$\II=\{(c,c+\delta)\, ; \; c\in\LL, \; \delta\in\Ni\}.$$

The global rule is a function $\function{\Phi}{\EE\times\II}{\EE}$. It takes in argument a configuration and a pair of cells to update, and maps it to the configuration that represents the next state of the system. The image $y=\Phi(x,(c,c'))$ is defined by:
$$(y_c,y_{c'})=\phi\big((x_{c+k},k\in \Np),(x_{c'+k},k\in \Np)\big) \mbox{ and for } d\notin\{c,c'\}, y_d=x_d,$$
where $\function{\phi}{\QQ^{\Np}\times\QQ^{\Np}}{\QQ^2}$ is the local rule that gives the new states of the pair of cells as a function of the states of their perception neighbourhood.

This rule is conservative if the image $y$ always satisfies $(y_c,y_{c'})=(x_c,x_{c'})$ or $(y_c,y_{c'})=(x_{c'},x_c)$.

Let $(u_t)_{t\in\N} \in \II^\N$ be a sequence of interacting pairs (in the following, the $u_t$ are chosen uniformly at random independently in $\II$). Starting from an initial condition $x \in \EE $, the system will evolve according to the sequence of states (or orbit) $ (x^t)_{t\geq 0}$ defined by $ x ^ 0 = x$ and $ x^{t+1}=\Phi(x^t, u_t) \mbox{ for any } t\geq 0$.

Given an IPS rule $\Phi$, we say that a set $A\subset \EE$ is an \emph{absorbing set} if: 
$\forall x\in A, \forall u\in\II, \Phi(x,u)\in A.$
We say that $A$ is \emph{reachable from any configuration} if: 
$$\forall x^0\in A, \; \exists T\geq 0, \; \exists (u_t)_{1\leq t\leq T} \in \II^T, \; x^T\in A.$$
We say that $A$ is a \emph{sink} if it is an absorbing set, which is reachable from any configuration.

In terms of IPS, we say that $\Phi$ is a solution to the strict spacing problem if it is a conservative IPS such that the set $\A$ of archipelagos is a sink.

\subsection{No solution to the strict spacing problem}\label{sec-nosol}

In order to solve the spacing problem, an idea is to design a rule, such that its evolution would result in decreasing the \emph{energy} of the configuration, that is, the number of adjacent cells in same state. This idea will be used to propose an approximate solution in Section~\ref{sec:glauber}. 
However, the next proposition proves that this idea does not allow us to solve the strict spacing problem. This is due to the existence of configurations that are not archipelagos but for which each cell can ``believe'' that it is part of an archipelago (by looking at the cells located within a finite range), see Fig.~\ref{fig:exCW}.

\begin{proposition}\label{prop:existence}
There is no IPS solution to the strict spacing problem.
\end{proposition}

\ProofPropExistence


\subsection{An IPS that synchronises checkerboards}
\label{sec:chckbdSynch}

We have seen that there is no IPS solution ensuring that once we have reached any archipelago configuration, we will remain in the set of archipelago configurations. At first sight, this could seem that there exist no solution to the spacing problem at all. However, to our surprise, we could notice that the loose problem is solvable. In fact, the {\em loose} problem is more demanding on the set of configurations that we do not leave. In other words, there might exist a subset $\A'$ of archipelagos such that once the configuration reaches $\A'$, it remains in it.


It can be observed that in Figure~\ref{fig:exCW}, the problem comes from the fact that there are two checkerboards of different phases. If we were able to synchronize these two checkerboards, we would reach one of the two perfect checkerboards (since both $\qO$ and $\qX$ have density $1/2$).

\newcommand{\Ce}[1]{C\und{e}^{#1}}
\newcommand{\Co}[1]{C\und{o}^{#1}}
For $q\in\QQ$, we denote:
\begin{align*}
\Ce{q}=\{x\in \EE\,;\; x_{i,j}=q \implies i+j \mbox{ is even}\},\\
\Co{q}=\{x\in \EE\,;\; x_{i,j}=q \implies i+j \mbox{ is odd}\}.
\end{align*}
These sets correspond to ``sub-checkerboards'' in the sense that in $\Ce{q}$ (resp. $\Co{q}$), state $q$ is the minority state and only appears on even (resp. odd) cells. We also introduce $C=\Ce{\qO}\cup \Ce{\qX}\cup \Co{\qO}\cup \Co{\qX}$. It is a subset of $\A$.

\begin{proposition}
\label{prop:spacing:loose:solution}
There is an IPS solution to the loose spacing problem for an even-sized grid. Indeed, the rule $\PhiC$ defined below is a conservative IPS having the property that the set $C$ is a sink.
\end{proposition}


To define $\PhiC$ with a local description, we introduce its interaction neighbourhood $\Ni=\{-1,0,1\}^2\setminus\{(0,0)\}$, that is, we allow interactions between a cell and its eight nearest-neighbours. The perception neighbourhood corresponds to von Neumann neighbourhood, that is, $\Np=\{(0,0),(0,1)(-1,0),(0,-1),(1,0)\}$.

Let $\II_4$ be the set of pairs of adjacent cells and $\DD_4$ be the set of diagonal pairs. The interaction set is $\II=\II_4\cup\DD_4$. The rule will act differently on diagonal pairs on the one hand, and horizontal and vertical pairs on the other hand.

For $(i,j)\in\II$, let $ \function{\tau_{(i,j)}}{\EE}{\EE}$ be the function that exchanges the states $x_i$ and $x_j$ in configuration $x$. Precisely, $\tau(x,{(i,j)})$ is defined with:
$$\tau(x,{(i,j)})_k = 
\begin{cases}
x_{j} & \TextIf k=i\\
x_{i} & \TextIf k=j\\
x_k & \TextIf k \notin \set{i,j}.\\
\end{cases}
$$

For a pair $(i,j)\in\II_4$, we define 
$$\Phi\und{C}(x,{(i,j)})=
\begin{cases}
\tau(x,{(i,j)}) \TextIf \mbox{ both cells } i \mbox{ and } j \mbox{ are not isolated,}\\
x \mbox{ otherwise.}\\
\end{cases}
$$

For a pair $(i,j)\in\DD_4$, we define $\PhiC(x,{(i,j)})=\tau(x,{(i,j)}).$

Note that if $x_i=x_j$, both cases above result in leaving $x$ unchanged.

In practice, numerical simulations show that in order to improve the speed of convergence to the archipelago, it is more appropriate to do the exchanges with different probability rates, depending on the state of the neighbourhood of the pair. However, these parameter do not affect the reachability properties, this is why we will here work with the simplest version of the model.

\figProof
\ProofSpacingLoose


\subsection{Glauber dynamics}\label{sec:glauber}

We will now present a family of stochastic IPS having the property to solve the spacing problem ``with an arbitrary precision'', by converging to a distribution on configurations (with same density as the initial configuration) for which the ``energy'' can be controlled. Let us thus precise this notion of energy.

The energy $E(x)$ of a configuration $x$ is the number of (horizontal or vertical) pairs $\qX\qX$ plus the number of (horizontal or vertical) pairs $\qO\qO$ in the configuration. Precisely, for $x\in\EE$,
$E(x)=\sum_{(i,j)\in\II_4} \ind[x_i=x_j],$
where we recall that $\II_4$ is the set of pairs of adjacent cells.

Let $\EE_k$  be the set of configurations of $\EE$ which contain $k$ cells in state $\qX$.

\begin{lemma}\label{lem:ens} If the grid is even-sized, then the configurations of $\EE_k$ of minimal energy are exactly the configurations of $\EE_k\cap \A=\EE_k\cap\A_q$, where $q$ is the minority state.
\end{lemma}

\ProofLemmaEns

For $(i,j)\in\II_4$, let us also define the local energy $E_{(i,j)}(x)$ of configuration $x$ at edge $(i,j)$ by:
$$E_{(i,j)}(x)=\sum_{(k,\ell)\in\V_4(i,j)}\ind[x_k=x_{\ell}],$$
where $\V_4(i,j)$ denotes the six edges of $\II_4\setminus\{(i,j)\}$ sharing a vertex with $(i,j)$.

Let $\beta\in\mathbb R$ be some fixed parameter. We propose the stochastic IPS dynamics defined as follows.
\begin{enumerate}
\item Choose uniformly at random a pair $u=(i,j)\in \II_4$ of horizontal or vertical consecutive cells.
\item Then, exchange the states of cells $i$ and $j$ with probability $p(x,u)$ defined by $$p(x,u)={\exp(\beta E_{u}(x)) \over \exp(\beta E_{u}(x))+\exp(\beta (6-E_{u}(x)))}={1\over 1+\exp(\beta (6-2E_{u}(x)))}.$$
\end{enumerate}

The number of cells in state $\qX$ is conserved by this dynamics, so that for any $k\in\{0,N_{\LL}\}$, it defines a discrete time Markov chain on $\EE_k$. The sequence of edges that will be chosen at each time step is given by a sequence of i.i.d. random variables $ (u_t)_{t\in\N} \in \II_4^\N $, where $ u_t $ is uniformly distributed on $\II_4$. Starting from an initial condition $x \in \EE_k$, the system evolves according to the sequence of states $(x^t)_{t\in \N}$ defined by $ x ^ 0 = x $ and 
$$x^{t+1}=\begin{cases}
\tau(x^t,u_t) & \mbox{ with probability } p(x,u_t),\\
x^t & \mbox{ with probability } 1-p(x,u_t).
\end{cases}
$$

This Markov chain is clearly irreducible and aperiodic. We denote its transition kernel by $P$. In particular, if $x\not=y$ and $y=\tau(x,u)$ for some $u\in\II_4$, then we have:
$P(x,y)={1\over\card \II_4} \, p(x,u).$

\begin{proposition} The Markov chain defined above is reversible, and its stationary distribution on $\EE_k$ is given by 
$ \mu_{\beta}(x)={1\over Z_{\beta}} \exp(-\beta E(x))$ for all $ x \in \EE_k$,
where $Z_{\beta}=\sum_{x\in\EE_k}\exp(-\beta E(x)).$

\end{proposition}

\begin{proof} Let us check that the detailed balance $\mu_{\beta}(x)P(x,y)=\mu_{\beta}(y)P(y,x)$ holds for any two configurations $x$ and $y$.
It is enough to prove that if $x\not=y$ and $y=\tau(x,u)$ for some $u\in\II_4$, then $\exp(-\beta E(x))p(x,u)=\exp(-\beta E(y))p(y,u)$.
But in that case, $E_u(y)=6-E_u(x)$ and $E(y)-E_u(y)=E(x)-E_u(x)$, so that the equality is satisfied.
\QED\end{proof}

\begin{proposition} If the grid is even-sized, then when $\beta\rightarrow\infty$, the distributions $\mu_{\beta}$ converge to the uniform measure on configurations of minimal energy, that is to the uniform measure on $\EE_k\cap \A$.
\end{proposition}

\begin{proof} It follows from the definition of $\mu_{\beta}$, and Lemma~\ref{lem:ens} above. \QED\end{proof}

To sum up, the Glauber dynamics gives a simple way to approach our goal of spacing out particles. Compared to our checkerboard synchronisation rule $ \PhiC$, it has the advantage of being simple and to use only horizontal and vertical interactions between cells. 
The distribution at the equilibrium can be determined analytically: it has same weight on all archipelagos with same number of particles, and the weight of non-archipelagos decreases exponentially as a function of $\beta$. But, as for $ \Phi_C$, there is still the need to know what are the time scales for observing the convergence to the equilibrium: a rule that would converge with a speed that is exponentially slow with the grid size would be useless in practice.

\section{The density classification problem on finite lattices}
\label{sec:densClassifRule}

As a first step, we propose to study here only how to use a modified version of $ \PhiC $ to solve the density  classification problem, leaving the Glauber dynamics for future work.
We build our solution to the density classification problem by combining $ \PhiC$ with the majority rule. 

For the checkerboard synchronisation dynamics, we introduce the following parametric variant of $\PhiC$, in order to increase the speed of convergence.
For a configuration $x\in\EE$ and a pair $(i,j)\in \II_4$, the new rule $ \PhiT $ is defined as follows:
(a) if both $i$ and $j$ have exactly one adjacent cell in the same state, the exchange (which is then always allowed in $\PhiC$) is now applied with a probability $\lambda$ ; (b) in all other cases, we apply $\PhiC$ ; (c) for a pair $(i,j)\in \DD_4$, the exchange (which is always allowed in $\PhiC$), is now done only with a probability $\chi$. For $\lambda=\chi=1$, we recover $\PhiC$.


We now combine $\PhiT$ with a majority rule, to obtain a rule $\PhiD$ defined as follows: 
for a configuration $x\in\EE$ and a pair $(i,j)=u\in\II$, 
$$\PhiD(x,u)=
\begin{cases}
\PhiT(x,u) \mbox{ with probability } 1-\epsilon,\\
{\mbox{Maj}}(x,i) \mbox{ with probability } \epsilon/2,\\
{\mbox{Maj}}(x,j) \mbox{ with probability } \epsilon/2,
\end{cases}
$$
where $\function{\mbox{Maj}}{\EE\times\LL}{\EE}$ is the function such that $y=\mbox{Maj}(x,i)$ is defined by $y_c=x_c$ if $c\neq i$ and $y_i$ is the majority state in the Moore neighbourhood of $c$ (the 8 nearest neighbours of $ c$). We can now state our main proposition.


\begin{proposition}
For an even-sized grid $ \LL = (\Z/a\Z)\times(\Z/b\Z)$ with $a,b \in 2\Z$, 
for any configuration $ x \in \LL $, and any non-zero value of $ \lambda $ and $ \chi $, 
the probability that $ \PhiD $ provides a good classification of $ x $ tends to 1 as $ \epsilon $ tends to 0.
\end{proposition}

\begin{proof}[Sketch]
The proof is the same as the one given for a one-dimensional system~\cite{FatTocs13}. 
For the sake of simplicity we can set $ \lambda = 1$ and $ \chi = 1$, that is, make $ \PhiT $ and $ \PhiC $ equal.
First, remark that despite the stochastic nature of our systems, archipelagos are well classified with probability one. 
In particular, if $ q $ is the minority state of $ x $, 
then the sets of sub-checkerboards $ \Ce{q} $ and $ \Co{q} $  are stable by the application of both $ \PhiT $ and the majority rule. 
Moreover, the majority rule either leaves $ x $ unchanged or diminishes by 1 the number of $ q $'s in $ x $.
In other words, the system can only converge to the right fixed point $ \qO^\LL $ or $ \qX^\LL$.

The second property to remark is that as $ \epsilon $ gets smaller, the probability of {\em not} applying the majority rule during the first $ k $ time steps tends to 1 for every value of $ k $. 
In other words, for a configuration $ x$ with minority symbol $ q $, we can make the probability to reach a sub-checkerboard in $ \Ce{q} \cup \Co{q} $ before applying the majority rule as high as needed.
\QED\end{proof}

In order to evaluate the quality of the rule in practice, let us now briefly explore how the two rules $ \PhiT $ and $ \PhiD$ behave with respect to their various settings.

\figHomoPhase

For the particle spacing problem, our rule is defined with three parameters:  the grid width $ L $ and the two probabilities of exchange $ \lambda $ and $ \chi $.

Let us first examine how to set $ \lambda $ and $ \chi $. For each setting of the system, we repeated $ 1000 $ experiments consisting of initializing the system with an independent Bernoulli of parameter $1/2$ for each cell and measuring the time needed to attain an archipelago configuration. Note that for the sake of making fair comparisons, we show here the {\em rescaled time}, that is, a time step is taken as $ L^2 $ random updates of the global rule.
We also take $ L $ to be even as for odd-sized grids there are initial conditions for which the particle spacing problem has no solution. In the previous studies, typical values of $ L $ were taken around $ 20 $, see e.g. Ref.~\cite{Oli14,FatDCP12} and references therein. 

Figure~\ref{fig:homo}-left shows how the average time to convergence varies as a function of $ \lambda$ for the specific value $ \chi=0.1$. 
It can be observed that for $ \lambda$ smaller than $ 0.4 $ the time is relatively small, while for higher values of $ \lambda $, the time increases drastically. 

To examine the effect of $ \chi$, we arbitrarily fixed the value of $ \lambda $ to $ 0.25$ and measured the average convergence time to an archipelago. 
Figure~\ref{fig:homo}-right shows how the average time to convergence varies as a function of $ \chi$.
Here again, it can be observed that for $ \chi$ smaller than $ 0.3 $ the time is relatively small, while for higher values of $ \chi $, the time increases drastically. 

Interestingly, these two experiments show that in order to avoid the existence of non-archipelago fixed-point configurations, the two parameters have to be set strictly greater than zero, but can not be set too high. It is an open question to determine if there exists a phase transition with respect to the convergence to an archipelago. This would imply that, for infinite systems, if $\lambda$ and $ \chi $ are set above a given threshold, with a high probability, the system does not converge to an archipelago.

Our third experiment is to observe the density classification itself.
For each random sample, we took random initial conditions with a uniform probability to be $ \qO $ or $ \qX $ for each cell independently. For even-sized grids, in case of equality between the number of $ \qO$s and $ \qX$s, we dismissed this initial condition and re-sampled another one with the same random distribution. 
We define the {\em quality} as the ratio of successful classifications, that is, the convergence to right fixed point $ \qO^\LL $ or $\qX^\LL$ depending on whether the initial condition has a density smaller than or greater than 1/2.  

\figK

Fig.~\ref{fig:koupa} shows the evolution of the quality as a function of the grid width~$ L $, for the particular setting $ \lambda =0.25$, $ \chi = 0.1$, $ \epsilon = 0.001$. (Surprisingly, we empirically remarked that the quality for $ \chi = 0.2$ is slightly lower.)
These results show that this rule has a quality that is comparable to the best two-dimensional classification rules known so far~\cite{Oli14,FatDCP12}.
Without surprise, for even-sized grids, the quality decreases as $ L $ gets larger, as it becomes more difficult to discriminate between the configurations that have approximately the same number of $ \qO$s and $\qX$s.
The curve for odd-sized grids is more surprising. Indeed, it shows an {\em increase} of the quality with $ L $, at least in the range of sizes that were examined.
We believe that this phenomenon results from the impossibility to space out the particles for some configurations of odd-sized grid. The system is in some sense ``forced'' to converge to a non-perfect configuration, in which the majority rule may introduce errors and make the system shift towards the wrong fixed point. However, as the size further increases this effect is less important and it is probable that for a given $\epsilon$, the difference between even-sized and odd-sized grids disappears.

\section{Some questions}

The density classification problem and the particle spacing problem can both be extended to infinite lattices. 
The set of configurations is then $\EE=\QQ^{\Z^d}$. 
For the density classification problem, 
a possible extension to infinite lattices consists in designing a cellular automaton on $\EE$ such that if the initial configuration is drawn independently for each cell according to a Bernoulli law of parameter $ p $, then if $p<1/2$, the density of $\qX$s converges to $0$, while if $p>1/2$, the density of $\qO$s converges to $0$.

This problem has already been studied by Marcovici and her collaborators~\cite{BFMM13}. In particular, it was shown that there is a simple example of deterministic cellular automaton that classifies the density on $\Z^2$: Toom's rule, which is the majority rule on the neighbourhood $\Neighb=\{(0,0),(0,1),(1,0)\}$. However, in dimension $1$, it is an open problem whether there exists a (deterministic or probabilistic) rule that classifies the density. Taati has partially answered this question by giving an argument that holds for densities close to zero or to one~\cite{Taa15}.


Similarly, for the particle spacing problem, an extension to infinite lattices consists in asking to design a conservative 
cellular automaton on $\EE$ such that if the initial configuration is drawn independently for each cell according to a Bernoulli law of parameter $ p $, then if $p<1/2$, the density of non-isolated $\qX$s converges to $0$, while if $p>1/2$, the density of non-isolated $\qO$s converges to $0$.
It is known that the traffic cellular automaton $F_{184}$ is a solution to that problem on $\Z$, but the problem remains open in dimension $d\geq 2$~\cite{BF2005}.

The IPS models we have introduced in Section~\ref{sec:IPS} are also interesting when studying them on $\Z^2$ instead of finite lattices. In that case, to define properly the model, we need to consider continuous-time updates: each interacting pair of $\II$ possess a clock that rings at times that are exponentially distributed (independently for the different pairs), and the local rule is applied when the clock rings. It is an open problem to know if there is a proper setting of the checkerboard synchronisation dynamics having the property to space particles on $\Z^2$. 
Another interesting model is the Glauber dynamics for $\beta=\infty$. In that case, we allow exchanges between cells only if it makes the energy decrease (the exchange is made with probability $1/2$ if the exchange does not change the value of the energy). On finite grids, this IPS has many fixed points that are not archipelagos, but starting from a configuration on $\Z^2$ drawn according to a Bernoulli measure, the behaviour could be different.


As far as the performance of the models is concerned, we can ask what are the best settings to obtain a good trade-off between the quality of classification and the time needed to converge to a fixed point.


We also ask if we can transform our IPS into probabilistic cellular automata for solving the two-dimensional density classification problem. This can be done by using more states or by sharing the randomness of the cells (see Ref.~\cite{MaiMar14} and references therein), but it is an open problem whether there is a solution within the usual framework of binary probabilistic cellular automata.